\documentclass{article}
\usepackage{amsfonts}
\usepackage{mathrsfs}
\usepackage{amsthm}
\usepackage{amssymb}
\usepackage{amsmath}
\usepackage{enumerate}
\usepackage{braket}
\usepackage{booktabs}

\theoremstyle{plain}
\newtheorem{theorem}{Theorem}[section]
\newtheorem{proposition}[theorem]{Proposition}
\newtheorem{corollary}[theorem]{Corollary}
\newtheorem{lemma}[theorem]{Lemma}
\newtheorem{conjecture}{Conjecture}[section]
\theoremstyle{definition}
\newtheorem{definition}{Definition}[section]

\newtheoremstyle{myremark}
  {3pt}
  {3pt}
  {\itshape}
  {}
  {\itshape}
  {.}
  {.5em}
  {}

\theoremstyle{myremark}
\newtheorem{remark}{Remark}[section]

\theoremstyle{definition}

\usepackage{chngcntr}
\counterwithout{table}{section}

\numberwithin{equation}{section}

\title{Spectral Criterion for Disorder-Free Localization of Quantum Walks on Hypercube: QBN Approach}

\author{Ce Wang\thanks{email: cewangwhu@163.com; orcid:0000-0002-2916-8478} \\
Shanghai Institute for Mathematics and Interdisciplinary Sciences\\
Shanghai 200433, People's Republic of China}
\date{}
\begin{document}
\maketitle

\begin{abstract}
We study disorder-free localization in a spin-coupled quantum walk on the hypercube within the quantum Bernoulli noises framework. We derive a general criterion for disorder-free localization: it occurs if and only if some spectral subspace of a reduced unitary evolution contains a vector whose position marginal probability is non-uniform. For the Grover walk with distance-dependent coupling \(\phi_\sigma=|\sigma|\pi/(n+1)\) in dimension \(d=4\), we explicitly construct a non-uniform fixed point in the all \(+1\) spin configuration. This implies that disorder-free localization occurs in our model. 
\end{abstract}

\vskip 2mm

\vskip 2mm


\section{Introduction}

Quantum walks, the quantum mechanical counterpart of classical random walks, have emerged as a powerful framework for quantum simulation and quantum computation over the past three decades. Since their introduction by \cite{Aharonov1993}, quantum walks have been shown to constitute a universal model of quantum computation \cite{VenegasAndraca2012} and provide exponential speedups in algorithmic applications \cite{Childs2003}. Among the various graph topologies studied, the hypercube has attracted particular attention due to its high symmetry, rich spectral properties, and relevance to quantum search algorithms and fault-tolerant quantum computing \cite{MooreRussell, Kempe2005, AlagicRussell}.

Anderson localization \cite{Anderson1958} have traditionally been understood as disorder-driven phenomena\cite{AizenmanMolchanov1993, FrohlichSpencer}. This raises the question whether disorder is truly necessary for localization. Smith et al. \cite{Smith2017} introduced the concept of disorder-free localization (DFL), showing that a translationally invariant system can dynamically generate effective disorder through an extensive set of conserved quantities. Subsequent work revealed diverse DFL mechanisms, including gauge invariance induced localization \cite{Brenes2018}, hidden Anderson localization in disorder-free Ising-Kondo lattices \cite{Yang2020}, disorder-free localization transitions in two-dimensional lattice gauge theories \cite{Chakraborty2022}, discrete-time quantum walks with onsite spin conserved quantities \cite{Danaci2021}, continuous-time quantum walks with permutation symmetry \cite{Balachandran2024}, and static potential engineering in hypercube networks \cite{Arkhipov2025}. These works collectively indicate that disorder is not a necessary condition for localization, and that DFL can emerge in a wide range of physical systems through distinct mechanisms.

In this work, we study a spin-coupled quantum walk on the hypercube within the quantum Bernoulli noises (QBN) framework introduced by Wang \cite{WCL-JMP-2010}. The walker interacts with a local spin environment through a position-dependent phase coupling, which preserves unitarity while making the effective coin operators position-dependent. Our main contribution is as follows.
\begin{itemize}
    \item  We derive a general criterion for disorder-free localization in our model. The total Hilbert space decomposes into invariant sectors labeled by spin configurations $\mathbf s$, and in each sector the evolution reduces to a unitary operator $\tilde W_{\mathbf s}$. By analyzing the time-averaged position probability distribution, we show that DFL occurs if and only if there exist a spin configuration $\mathbf s$ and an eigenvalue $\lambda$ of $\tilde W_{\mathbf s}$ such that the corresponding eigenspace contains a vector whose position marginal probability is not uniform. This criterion reduces the problem of DFL to the study of fixed-point equations when $\lambda=1$, or more generally to spectral subspaces of $\tilde W_{\mathbf s}$.
    \item We apply this criterion to the Grover walk with distance-dependent coupling
\[
\phi_\sigma = \frac{|\sigma|}{n+1}\pi,
\]
where $|\sigma|$ denotes the Hamming weight of the vertex $\sigma$ in the hypercube $\Gamma_n$. For dimension $d=4$, we explicitly construct a non-uniform fixed point in a fixed spin sector, thereby proving that DFL occurs. It shows that even in a highly symmetric graph such as the hypercube, coupling to a local spin environment with carefully chosen position-dependent phases can localize the walker, despite the absence of disorder.
\end{itemize}

Our results complement existing DFL mechanisms. In the present model, DFL arises from a purely phase-driven mechanism within a fixed spin sector: the spin environment generates an effective position-dependent phase, which breaks translational symmetry in that sector and thereby enables localization on the hypercube. This is consistent with the general picture that some form of symmetry breaking is needed to turn conserved quantities into effective disorder. However, our rigorous proof of DFL is at present restricted to the explicit \(d=4\) construction. We also observed the same consequence for \(d=6\) and \(d=8\) through numerical computations. This suggests a possible generalization to all even dimensions, but a complete proof remains open.

The paper is organized as follows. In Sec.~2 we recall the QBN framework. Section 3 we introduce the spin-coupled hypercube walk and derive the reduced dynamics and path-integral representation for a fixed spin configuration. In Sec.~4 we establish the general DFL criterion and apply it to the Grover walk with distance-dependent coupling, constructing a non-uniform fixed point for $d=4$. We conclude in Sec.~5 with a discussion.

\section{Preliminaries}

We first recall some basic preliminaries, for more details, see \cite{WCL-JMP-2010} and \cite{Wang2022}. Let $\Omega = \{-1,1\}^{\mathbb{N}}$ be the Bernoulli space equipped with the canonical product measure $\mathbb{P}$. The space of square-integrable Bernoulli functionals is denoted by $\mathfrak{h} = L^2(\Omega, \mathcal{F}, \mathbb{P})$. Let $\Gamma$ be the finite power set of $\mathbb{N}$. It is known that $\mathfrak{h}$ has an orthonormal basis (ONB)
$\{Z_{\sigma}\mid \sigma\in \Gamma\}$ of the form: $Z_{\emptyset}=1$ and
\begin{equation}\label{h_ONB}
    Z_{\sigma} = \prod_{j\in \sigma}Z_j,\quad \text{$\sigma \in \Gamma$, $\sigma \neq \emptyset$},
\end{equation}
where $Z=(Z_n)_{n\geq 0}$ is some sequence of Bernoulli functionals. For each $k \in \mathbb{N}$, the annihilation operator $\partial_k$ and the creation operator $\partial_k^*$ act on the ONB $\{Z_\sigma : \sigma \subset \mathbb{N},\ |\sigma| < \infty\}$ as
\begin{equation}
\partial_k Z_\sigma = \mathbf{1}_\sigma(k) Z_{\sigma\setminus k}, \qquad 
\partial_k^* Z_\sigma = (1 - \mathbf{1}_\sigma(k)) Z_{\sigma\cup k}.
\end{equation}
These operators satisfy the canonical anti-commutation relations (CAR)
\begin{equation}
\partial_k^2 = (\partial_k^*)^2 = 0, \qquad 
\partial_k\partial_k^* + \partial_k^*\partial_k = I,
\end{equation}
and $[\partial_j, \partial_k] = 0$ for all $j \neq k$.
In the literature, the operator family $\{\partial_k,\, \partial_k^*\mid k\in \mathbb{N}\}$
is known as quantum Bernoulli noises (QBN)

We denote by $\mathfrak{h}_n$ the subspace of $\mathfrak{h}$ spanned by $\{Z_{\sigma} \mid \sigma \in \Gamma_n\}$, namely
\[
  \mathfrak{h}_n = \mathrm{span}\{Z_{\sigma} \mid \sigma \in \Gamma_n\},
\]
where $\Gamma_n$ is the power set of $\mathbb{N}_n=\{0, 1, \cdots, n\}$. Note that $\Gamma_n\subset \Gamma$ and $\#(\Gamma_n)=2^{n+1}$, which implies that $\mathfrak{h}_n$ is a $2^{n+1}$-dimensional closed subspace of $\mathfrak{h}$ and for all $k\in \mathbb{N}_n$, $\partial_k$ and $\partial_k^*$
leave $\mathfrak{h}_n$ invariant.

It is known that the $n$-dimensional hypercube $(V^{(n)}, E^{(n)})$ is a regular graph and its degree is exactly $n$. We give an alternative description of a hypercube. For a nonnegative integer $n$, let $\mathbb{N}_n = \{0,1,\dots,n\}$. Consider $\Gamma_n$, two elements $\sigma$ and $\tau$ in $\Gamma_n$ are said to be adjacent if $\#(\sigma\triangle \tau)=1$. In that case, we writes $\sigma\sim \tau$. We denote by $(\Gamma_n, \mathfrak{E}_n)$ the graph with $\Gamma_n$ being the vertex set and $\mathfrak{E}_n$ being the edge set. $\mathfrak{E}_n$ is given by
\[
 \mathfrak{E}_n=\big\{\, \{\sigma, \tau\} \mid \sigma, \tau\in \Gamma_n,\, \#(\sigma\bigtriangleup \tau)=1\,\big\},
\]
where $\sigma\bigtriangleup \tau = (\sigma\setminus \tau) \cup (\tau\setminus\sigma)$.

\begin{lemma}
Let $n\geq 0$ be a nonnegative integer. Then the graph $(\Gamma_n, \mathfrak{E}_n)$ is isomorphic to the $(n+1)$-dimensional hypercube $(V^{(n+1)}, E^{(n+1)})$.
\end{lemma}

\begin{lemma}
Consider the graph $(\Gamma_n,\mathfrak{E}_n)$. Let $\sigma$, $\tau\in \Gamma_n$ be its vertices.
Then $\sigma\sim \tau$ if and only if there exists a unique $k\in \mathbb{N}_n$ such that
\begin{equation}\label{eq}
  (\partial_k^*+\partial_k)Z_{\sigma} = Z_{\tau}.
\end{equation}
\end{lemma}

Let $\sigma$ be a vertex in the graph $(\Gamma_n,\mathfrak{E}_n)$ and $k\in \mathbb{N}_n$. Then, using properties of $\partial_k$ and $\partial_k^*$, we have
\begin{equation*}
  (\partial_k^*+\partial_k)Z_{\sigma}
= \left\{
    \begin{array}{ll}
      Z_{\sigma\setminus k}, & \hbox{$k\in \sigma$;}\\
      Z_{\sigma\cup k}, & \hbox{$k\notin \sigma$.}
    \end{array}
  \right.
\end{equation*}
On the hand, $\sigma\sim (\sigma\setminus k)$ when $k\in \sigma$; $\sigma\sim (\sigma\cup k)$ when $k\notin \sigma$.
Therefore, the operator $(\partial_k^*+\partial_k)$ on $\mathfrak{h}_n$ behaves actually as a shift operator. Thus we define 
\begin{equation}\label{shift_operator}
  \Xi_k:=\partial_k+\partial_k^*,\quad k\in \mathbb{N}_n
\end{equation} 
as the shift operator of the quantum walk on the hypercube. For $k\in \mathbb{N}_n$, $\Xi_k$ is a self-adjoint unitary operator, the operator sequence $\{\Xi_k\mid k\in\mathbb{N}_n\}$ is commuting.

A coin operator system on a separable complex Hilbert space $\mathcal{X}$ of dimension $d_{\mathcal{X}} \ge n+1$ is a set $\{C_k : 0 \le k \le n\}$ of bounded operators satisfying the orthogonality and completeness conditions
\begin{equation}
C_j^* C_k = 0 \quad (j \neq k), \qquad 
\sum_{k=0}^n C_k^* C_k = I_{\mathcal{X}}.
\end{equation}

\begin{lemma}\label{coin-characterization}\cite{Wang2022}
Let\ $\mathfrak{C}=\{ C_k \mid 0\leq k\leq n\}$ be a system of bounded operators on $\mathcal{X}$. Then the following statements are equivalent:
\begin{enumerate}
  \item[(1)] The system\ $\mathfrak{C}=\{ C_k \mid 0\leq k\leq n\}$ is a coin operator system on $\mathcal{X}$.
  \item[(2)] There exist a unitary operator $U$ and a resolution of the identity $\{P_k \mid 0\leq k \leq n\}$ on $\mathcal{X}$ such that
   $C_k = P_kU$,\ \ $0\leq k\leq n$.
\end{enumerate}
\end{lemma}

We take $\mathfrak{h}_n$ and $\{\Xi_k \mid k\in\mathbb{N}_n\}$
as the position space and shift operators respectively of discrete-time quantum walk on $(\Gamma_n,\mathfrak{E}_n)$. Then the walk takes $\mathfrak{h}_n\otimes \mathcal{K}$ as its state space and the states of the walk are represented by unit vectors in $\mathfrak{h}_n\otimes \mathcal{K}$. The time evolution of the walk is governed by
\begin{equation}
W = \sum_{k=0}^n \Xi_k \otimes C_k,
\end{equation}
which defines a quantum walk on the $2^{n+1}$-dimensional hypercube.

\section{Spin-Induced Model and Reduced Dynamics}

We now introduce a spin-coupled quantum walk on the hypercube within the QBN framework. Let $\Gamma_n$ be the vertex set of the $(n+1)$-dimensional hypercube, with position space $\mathfrak{h}_n$ and coin space $\mathcal K$ as defined before. We now introduce a separable complex Hilbert space $\mathfrak{h}_M$ as the spin environment. The index set is chosen to be $\mathbb{N}_M = \{0,1,\dots,M\}$ with $M = 2^{n+1} - 1$, so that $\mathbb{N}_M$ has exactly the same cardinality as the vertex set $\Gamma_n$. The space $\mathfrak{h}_M$ is equipped with the orthonormal basis  $\{Z_\gamma : \gamma \subseteq \mathbb{N}_M\}$. Here and throughout, $Z_\sigma$ with $\sigma\subseteq\mathbb N_n$ denotes position-space basis vectors, while $Z_\gamma$ with $\gamma\subseteq\mathbb N_M$ denotes spin-space basis vectors. The total Hilbert space is therefore
\begin{equation}\label{eq:total_space}
\mathcal{H} := \mathfrak{h}_n \otimes \mathcal{K} \otimes \mathfrak{h}_M.
\end{equation}
The spin environment is coupled to the walker through a position-dependent spin rotation operator, which we construct below. To this end, we first define a spin flip operator for each vertex.

Let $o: \Gamma_n \to \mathbb{N}_M$ be a fixed bijection. For each vertex $\sigma \in \Gamma_n$, $o(\sigma) \in \mathbb{N}_M$. We define $\mathfrak{S}_\sigma$ and $\mathfrak{S}_\sigma^*$ on the spin space $\mathfrak{h}_M$ by:
\begin{equation}
\mathfrak{S}_\sigma Z_\gamma = \partial_{o(\sigma)}Z_\gamma = \mathbf{1}_\gamma(o(\sigma)) Z_{\gamma\setminus\{o(\sigma)\}},\qquad
\mathfrak{S}_\sigma^* Z_\gamma = \partial_{o(\sigma)}^*Z_\gamma = (1 - \mathbf{1}_\gamma(o(\sigma)) Z_{\gamma\cup\{o(\sigma)\}}.
\end{equation}
The operators $\mathfrak{S}_\sigma$ and $\mathfrak{S}_\sigma^*$ have the same algebraic structure as $\partial_k$ and $\partial_k^*$.

The spin flip operator is then defined as
\begin{equation}
\Pi_\sigma^{\mathrm{spin}} := \mathfrak{S}_\sigma^* + \mathfrak{S}_\sigma \in \mathcal{B}(\mathfrak{h}_M).
\end{equation}
Its action on the basis is
\begin{equation}
\Pi_\sigma^{\mathrm{spin}} Z_\gamma = Z_{\gamma \triangle \{o(\sigma)\}},
\end{equation}
which flips $o(\sigma)$ in $\gamma$. 

Define the position-dependent spin rotation operator on the total Hilbert space $\mathcal H = \mathfrak h_n \otimes \mathcal K \otimes \mathfrak h_M$ as
\begin{equation}
R := \sum_{\sigma \in \Gamma_n} |Z_\sigma\rangle\langle Z_\sigma| \otimes I_{\mathcal K} \otimes e^{-i\phi_\sigma \Pi_\sigma^{\mathrm{spin}}}.
\label{eq:spin_rotation}
\end{equation}
 
The full evolution operator \(\tilde W = W R\) has the explicit form
\begin{equation}
\tilde W = \sum_{k=0}^{n}\sum_{\sigma\in\Gamma_n}
|Z_{\sigma\triangle\{k\}}\rangle\langle Z_\sigma|
\otimes C_k \otimes 
e^{-i\phi_{\sigma}\Pi_\sigma^{\mathrm{spin}}}.
\label{eq:W_final}
\end{equation}

\begin{proposition}
\label{pro:W_unitary}
The operator $\tilde W$ defined in \eqref{eq:W_final} is unitary on $\mathcal{H}$.
\end{proposition}

\begin{proof}
The operator $W$ is unitary by the standard QBN quantum walk construction. For each $\sigma$, $e^{-i\phi_\sigma \Pi_\sigma^{\mathrm{spin}}}$ is unitary since $\Pi_\sigma^{\mathrm{spin}}$ is self-adjoint. Thus $R$ is unitary as a direct sum of unitaries over the orthogonal position subspaces. Therefore $\tilde W = W R$ is unitary.
\end{proof}

In translationally invariant systems, the mechanism for disorder-free localization replace the role of an external random potential. In our framework, this mechanism is provided by a family of conserved quantities arising from the spin environment. These conserved quantities decompose the total Hilbert space into an exponentially large number of invariant subspaces. 

In the following, any operator $A$ defined in one of the three factor spaces in \eqref{eq:total_space} is implicitly extended to $\mathcal{H}$ by $A \mapsto A \otimes I$ on the remaining factors, and we keep the same notation for the extended operator.

\begin{proposition}
\label{pro:conservation}
For each vertex $\sigma \in \Gamma_n$, the spin flip operator $\Pi_\sigma^{\mathrm{spin}}$ commutes with the evolution operator $W$:
\begin{equation}
[\Pi_\sigma^{\mathrm{spin}}, \tilde{W}] = 0,
\end{equation}
where \( [\Pi_\sigma^{\mathrm{spin}}, \tilde{W}] := \Pi_\sigma^{\mathrm{spin}}\tilde{W}-\tilde{W}\Pi_\sigma^{\mathrm{spin}}  \) denotes the commutator.
\end{proposition}

\begin{proof}
Under the convention established above, $\Pi_\sigma^{\mathrm{spin}}$ acts as $I_{\mathfrak{h}_n} \otimes I_{\mathcal{K}} \otimes \Pi_\sigma^{\mathrm{spin}}$ on the total Hilbert space. It therefore acts trivially on the position and coin factors. Consequently, it commutes with the position projection $|Z_{\sigma \triangle \{k\}}\rangle\langle Z_{\sigma}|$ and the operator $C_k$. Its commutation with the spin rotation:
\begin{equation}
[\Pi_\sigma^{\mathrm{spin}}, e^{-i\phi_{\tau} \Pi_\tau^{\mathrm{spin}}}] = 0,
\end{equation}
Therefore, $[\Pi_\sigma^{\mathrm{spin}}, \tilde{W}] = 0$ for all $\sigma$.
\end{proof}

Since all $\Pi_\sigma^{\mathrm{spin}}$ commute with each other, they admit a complete set of simultaneous eigenstates. These eigenstates are precisely the Hadamard basis states, labeled by the eigenvalue configurations $\mathbf{s}= (s_\sigma)_{\sigma \in \Gamma_n}$, where $s_\sigma = \pm 1$. For each configuration $\mathbf{s}$, we define the eigenspace
\begin{equation}
\mathcal{H}_{\mathbf{s}} := \left\{ \Psi \in \mathcal{H} \;\middle|\; \Pi_\sigma^{\mathrm{spin}} \Psi = s_\sigma \Psi,\ \forall \sigma \in \Gamma_n \right\}.
\end{equation}

\begin{proposition}
\label{pro:eigenspace_decomp}
The total Hilbert space decomposes into a direct sum of invariant eigenspaces:
\begin{equation}
\mathcal{H} = \bigoplus_{\mathbf{s} \in \{\pm 1\}^{\Gamma_n}} \mathcal{H}_{\mathbf{s}}.
\end{equation}
Each $\mathcal{H}_{\mathbf{s}}$ is invariant under the evolution operator $W$. The total number of eigenspaces is
\begin{equation}
\#\{\mathcal{H}_{\mathbf{s}}\} = 2^{|\Gamma_n|} = 2^{2^{n+1}}.
\end{equation}
\end{proposition}

\begin{proof}
Since the operators $\Pi_\sigma^{\mathrm{spin}}$ are self-adjoint and commute pairwise, the spectral theorem for commuting families of self-adjoint unitary operators gives the direct sum decomposition. The invariance of each $\mathcal{H}_{\mathbf{s}}$ under $\tilde{W}$ follows from proposition \ref{pro:conservation}. The number of eigenspaces equals the number of possible eigenvalue configurations, which is $2^{|\Gamma_n|}$.
\end{proof}

 Let $\mathcal{S}: \{\pm1\}^{\Gamma_n} \to \mathfrak{h}_M$ be the map defined by
\begin{equation}\label{eq_S(s)}
    \mathcal{S}(\mathbf{s}) := \bigotimes_{\sigma \in \Gamma_n} \mathbf{s}_\sigma^X,
\end{equation}
where $\mathbf{s} = (s_\sigma)_{\sigma\in\Gamma_n}$ is the eigenvalue configuration and for each $\sigma\in\Gamma_n$, $\mathbf{s}_\sigma^X$ is the eigenvector of $\Pi_\sigma^{\mathrm{spin}}$ with $s_\sigma$ be the related eigenvalue 
\[
\Pi_\sigma^{\mathrm{spin}} \mathbf{s}_\sigma^X = s_\sigma \mathbf{s}_\sigma^X.
\]
 Thus $\mathcal{S}(\mathbf{s})$ is the Hadamard basis state corresponding to configuration $\mathbf{s}$. The set $\{\mathcal{S}(\mathbf{s}) : \mathbf{s}\in\{\pm1\}^{\Gamma_n}\}$ forms a complete orthonormal basis of $\mathfrak{h}_M$ and the operators $\Pi_\sigma^{\mathrm{spin}}$ are diagonal in the Hadamard basis.

\begin{remark}
For each $\mathbf{s} = (s_\sigma)_{\sigma\in\Gamma_n}$, $\mathcal{S}(\mathbf{s})$ can be obtained from the basis $\{Z_\gamma : \gamma \subseteq \mathbb{N}_M\}$ via the Walsh--Hadamard transform. The Hadamard basis states corresponding to $\sigma$ are
\[
\mathbf{s}_\sigma^X = \frac{1}{\sqrt{2}}\left(Z_\emptyset + s_\sigma Z_{\{o(\sigma)\}}\right),\qquad s_\sigma = \pm1.
\]
$\mathcal{S}(\mathbf{s}) = \bigotimes_{\sigma\in\Gamma_n} \mathbf{s}_\sigma^X$ is then related to the basis by
\[
\mathcal{S}(\mathbf{s}) = \frac{1}{2^{|\Gamma_n|/2}}
\sum_{\gamma\subseteq \mathbb{N}_M}
\left(\prod_{\sigma\in o^{-1}(\gamma)} s_\sigma\right) Z_\gamma,
\]
where $o^{-1}(\gamma) = \{\sigma\in\Gamma_n : o(\sigma)\in\gamma\}$ is the set of vertices whose images lie in $\gamma$. 
\end{remark}

In a fixed configuration $\mathbf s$ and corresponding eigenspace $\mathcal{H}_{\mathbf{s}}$, the spin rotation operator reduces to a scalar:
\begin{equation}
e^{-i\phi_{\sigma} \Pi_\sigma^{\mathrm{spin}}} \mathcal{S}(\mathbf{s}) = e^{-i\phi_{\sigma} s_\sigma} \mathcal{S}(\mathbf{s}).
\end{equation}

\begin{proposition}[Reduced evolution operator]
\label{prop:reduced_evolution}
For each spin configuration \(\mathbf{s}\in\{\pm1\}^{\Gamma_n}\), the restriction of  \(\tilde W\) to the invariant subspace \(\mathcal H_{\mathbf{s}}\) is the unitary operator
\begin{equation}\label{eq:W_s_def}
\tilde W_{\mathbf{s}}
:=
\tilde W\big|_{\mathcal H_{\mathbf{s}}}
=
\sum_{k=0}^{n} \sum_{\sigma \in \Gamma_n}
|Z_{\sigma \triangle \{k\}}\rangle\langle Z_\sigma|
\otimes C_k\,
e^{-i\phi_{\sigma} s_\sigma}.
\end{equation}
\end{proposition}

\begin{proof}
The subspace \(\mathcal H_{\mathbf{s}}\) is invariant under \(\tilde W\) by the conservation of the spin-flip operators. In the basis \(Z_\sigma\otimes e_j\otimes \mathcal S(\mathbf{s})\), the spin rotation acts as multiplication by \(e^{-i\phi_\sigma s_\sigma}\), while the position-coin walk \(W\) acts as \(\sum_k \Xi_k\otimes C_k\). Therefore the restriction is exactly the formula above. Unitarity follows from the fact that \(\tilde W\) is unitary and \(\mathcal H_{\mathbf{s}}\) is invariant.
\end{proof}
Since \(\mathcal H = \bigoplus_{\mathbf{s}\in\{\pm1\}^{\Gamma_n}} \mathcal H_{\mathbf{s}}\), the evolution operator $\tilde W$ is just the direct sum of reduced operators:
\begin{equation}\label{eq:direct_sum_of_W}
\tilde W = \bigoplus_{\mathbf{s}\in\{\pm1\}^{\Gamma_n}} \tilde W_{\mathbf{s}}.
\end{equation}

Next, we will derive position probability distribution and its path-integral representation for an arbitrary initial position-coin state in a given spin configuration. 

Let \(\mathbf s=(s_\sigma)_{\sigma\in\Gamma_n}\in\{\pm1\}^{\Gamma_n}\) be a fixed spin configuration, and let \(\mathcal S(\mathbf s)\) be the corresponding eigenstate of the spin flip operators. Consider an initial total state of the form
\[
\Psi_0 = \Phi_0 \otimes \mathcal S(\mathbf s),
\]
where \(\Phi_0\in\mathfrak h_n\otimes\mathcal K\) is an arbitrary unit vector. Then $\Psi_0\in\mathcal H_{\mathbf s}$,
and the total state at time t is given by
\begin{equation}\label{eq:Psi_t}
\Psi_t = \tilde W^t \Psi_0 = \tilde W_{\mathbf s}^t \Phi_0 \otimes \mathcal S(\mathbf s),
\end{equation}
 Thus at time \(t\), the probability of finding the walker at position \(\sigma\) is
\begin{equation}\label{eq:prob_time_t}
P_t(\sigma|\Phi_0,\mathbf s)
=
\left\| Q_\sigma \, \tilde W_{\mathbf s}^t \Phi_0 \right\|^2,
\end{equation}
where 
\begin{equation}\label{pos_projection}
Q_\sigma := |Z_\sigma\rangle\langle Z_\sigma|\otimes I_{\mathcal K}.
\end{equation}
We now derive a path-integral formula for an arbitrary unit vector $\Phi_0 = \sum_{\tau\in\Gamma_n} Z_\tau\otimes u_\tau$ in the fixed $\mathbf s$. 

Define the set of paths of length \(t\) from \(\tau\) to \(\sigma\) as
\begin{equation}\label{eq:path_set}
\Lambda_t(\sigma,\tau) := \{(\sigma_0,\dots,\sigma_t) : \sigma_0=\tau,\ \sigma_t=\sigma,\ \sigma_j\sim\sigma_{j-1}\ \text{for all } j=1,\dots,t\}.
\end{equation}
For a path
\[
\gamma=(\sigma_0,\dots,\sigma_t)\in\Lambda_t(\sigma,\tau),
\]
let \(k_j\in\mathbb N_n\) be the unique direction such that \(\sigma_j=\sigma_{j-1}\triangle\{k_j\}\). 

Define the coin amplitude vector
\begin{equation}\label{coin_amplitude}
\mathcal A_\gamma := C_{k_t}\cdots C_{k_1} u_\tau\in\mathcal K.
\end{equation}
For each vertex \(\eta\in\Gamma_n\), set
\[
v_\eta(\gamma) := \#\{m\in\{0,\dots,t-1\}:\sigma_m=\eta\},
\qquad
d_\eta(\gamma,\gamma') := v_\eta(\gamma)-v_\eta(\gamma').
\]

\begin{theorem}
\label{thm:path_integral_fixed_s}
Let \(\Phi_0 = \sum_{\tau\in\Gamma_n} Z_\tau\otimes u_\tau\) be a unit vector in \(\mathfrak h_n\otimes\mathcal K\), and let \(\mathbf s\in\{\pm1\}^{\Gamma_n}\) be fixed. Then for any \(t\ge 0\) and \(\sigma\in\Gamma_n\),
\begin{equation}\label{eq:path_integral_fixed_s}
P_t(\sigma|\Phi_0,\mathbf s)
=
\sum_{\tau,\tau'\in\Gamma_n}
\sum_{\substack{\gamma\in\Lambda_t(\sigma,\tau)\\\gamma'\in\Lambda_t(\sigma,\tau')}}
\langle \mathcal A_\gamma,\mathcal A_{\gamma'}\rangle
\,
\prod_{\eta\in\Gamma_n}
\left(e^{-i\phi_\eta s_\eta}\right)^{d_\eta(\gamma,\gamma')},
\end{equation}
where the product is over all vertices \(\eta\in\Gamma_n\). In particular, when \(\Phi_0=Z_\tau\otimes u_0\) is localized at a single vertex, this reduces to
\begin{equation}\label{eq:path_integral_local}
P_t(\sigma|\tau,u_0,\mathbf s)
=
\sum_{\gamma,\gamma'\in\Lambda_t(\sigma,\tau)}
\langle \mathcal A_\gamma,\mathcal A_{\gamma'}\rangle
\,
\prod_{\eta\in\Gamma_n}
\left(e^{-i\phi_\eta s_\eta}\right)^{d_\eta(\gamma,\gamma')}.
\end{equation}
\end{theorem}

\begin{proof}
By linearity of \(\tilde W_{\mathbf s}^t\),
\[
Q_\sigma \tilde W_{\mathbf s}^t \Phi_0
=
\sum_{\tau\in\Gamma_n}
Q_\sigma \tilde W_{\mathbf s}^t (Z_\tau\otimes u_\tau).
\]
For each fixed \(\tau\), 
\[
Q_\sigma \tilde W_{\mathbf s}^t (Z_\tau\otimes u_\tau)
=
\sum_{\gamma=(\sigma_0,\dots,\sigma_t)\in\Lambda_t(\sigma,\tau)}
\left(\prod_{m=0}^{t-1} e^{-i\phi_{\sigma_m}s_{\sigma_m}}\right)
\mathcal A_\gamma,
\]
where the product is taken along the path \(\gamma\), and \(\sigma_0=\tau\), \(\sigma_t=\sigma\). Therefore
\[
Q_\sigma \tilde W_{\mathbf s}^t \Phi_0
=
\sum_{\tau\in\Gamma_n}
\sum_{\gamma\in\Lambda_t(\sigma,\tau)}
\left(\prod_{m=0}^{t-1} e^{-i\phi_{\sigma_m}s_{\sigma_m}}\right)
\mathcal A_\gamma.
\]
Taking the squared norm in \(\mathcal K\) and expanding the double sum over paths, we obtain
\[
P_t(\sigma|\Phi_0,\mathbf s)
=
\sum_{\tau,\tau'\in\Gamma_n}
\sum_{\substack{\gamma\in\Lambda_t(\sigma,\tau)\\\gamma'\in\Lambda_t(\sigma,\tau')}}
\langle \mathcal A_\gamma,\mathcal A_{\gamma'}\rangle
\,
\prod_{m=0}^{t-1} e^{-i\phi_{\sigma_m}s_{\sigma_m}}
\prod_{m'=0}^{t-1} e^{i\phi_{\sigma'_{m'}}s_{\sigma'_{m'}}}.
\]
Now group the factors according to the vertex \(\eta\in\Gamma_n\). Since each vertex \(\eta\) appears exactly \(v_\eta(\gamma)\) times among \(\sigma_0,\dots,\sigma_{t-1}\), the product over time steps can be rewritten as a product over vertices:
\[
\prod_{m=0}^{t-1} e^{-i\phi_{\sigma_m}s_{\sigma_m}}
\prod_{m'=0}^{t-1} e^{i\phi_{\sigma'_{m'}}s_{\sigma'_{m'}}}
=
\prod_{\eta\in\Gamma_n}
\left(e^{-i\phi_\eta s_\eta}\right)^{v_\eta(\gamma)-v_\eta(\gamma')}
=
\prod_{\eta\in\Gamma_n}
\left(e^{-i\phi_\eta s_\eta}\right)^{d_\eta(\gamma,\gamma')}.
\]
Substituting this into the expression for \(P_t(\sigma|\Phi_0,\mathbf s)\) gives the claimed formula. The localized case \(\Phi_0=Z_\tau\otimes u_0\) follows by setting all \(u_{\tau'}=0\) for \(\tau'\neq\tau\).
\end{proof}
\begin{remark}
It is interesting to note that if the initial spin state is the vacuum \(Z_\emptyset\), then the total initial state is
\[
\Psi_0 = \Phi_0 \otimes Z_\emptyset
= \frac{1}{2^{|\Gamma_n|/2}}
\sum_{\mathbf s\in\{\pm1\}^{\Gamma_n}}
\Phi_0 \otimes \mathcal S(\mathbf s).
\]
Consequently, the position probability distribution is the uniform average of the fixed-configuration probabilities:
\[
P_t(\sigma|\Phi_0,Z_\emptyset)
=
\frac{1}{2^{|\Gamma_n|}}
\sum_{\mathbf s\in\{\pm1\}^{\Gamma_n}}
P_t(\sigma|\Phi_0,\mathbf s).
\]
Substituting the path-integral formula from Theorem~\ref{thm:path_integral_fixed_s} and using the independence of the spin variables, we obtain
\begin{equation}\label{eq:path_integral_vaccum}
P_t(\sigma|\Phi_0,Z_\emptyset)
=
\sum_{\tau,\tau'\in\Gamma_n}
\sum_{\substack{\gamma\in\Lambda_t(\sigma,\tau)\\\gamma'\in\Lambda_t(\sigma,\tau')}}
\langle \mathcal A_\gamma,\mathcal A_{\gamma'}\rangle
\prod_{\eta\in\Gamma_n}
\cos\left(\phi_\eta\, d_\eta(\gamma,\gamma')\right).
\end{equation}
\end{remark}

\section{Time-Averaged probability distribution and Disorder-Free Localization}

We are interested in the long-time average of the position probability distribution for the fixed spin configuration \(\mathbf s\):
\begin{equation}\label{eq:time_average_s}
\bar P_\sigma(\Phi_0,\mathbf s)
:=
\lim_{T\to\infty}\frac{1}{T}\sum_{t=0}^{T-1}
P_t(\sigma|\Phi_0,\mathbf s).
\end{equation}
Using \eqref{eq:prob_time_t}, the time-average position probability in the fixed $\mathbf s$ is
\[
\bar P_\sigma(\Phi_0,\mathbf s)
=
\lim_{T\to\infty}\frac{1}{T}\sum_{t=0}^{T-1}
\left\|Q_\sigma \tilde W_{\mathbf s}^t \Phi_0\right\|^2.
\]

Recall that \(\tilde W_{\mathbf s}\) is a unitary operator, it admits a spectral decomposition
\[
\tilde W_{\mathbf s}
=
\sum_{\lambda\in E(\tilde W_{\mathbf s})}
\lambda \, \mathcal P_{\mathbf s,\lambda},
\]
where \(E(\tilde W_{\mathbf s})\) is the finite set of eigenvalues, and \(\mathcal P_{\mathbf s,\lambda}\) is the orthogonal projection onto the eigenspace associated with \(\lambda\). Since \(\tilde W_{\mathbf s}\) is unitary, $|\lambda|=1,\forall \lambda$, so \(\lambda = e^{i\theta_\lambda}\) for some \(\theta_\lambda\in[0,2\pi)\).

\begin{theorem}
\label{thm:time_average_spectral}
Let \(\Phi_0\in\mathfrak h_n\otimes\mathcal K\) be a unit vector and let \(\mathbf s\in\{\pm1\}^{\Gamma_n}\) be fixed. Then the time-averaged position probability is given by
\begin{equation}\label{eq:time_average_spec}
\bar P_\sigma(\Phi_0,\mathbf s)
=
\sum_{\lambda\in E(\tilde W_{\mathbf s})}
\left\| Q_\sigma\, \mathcal P_{\mathbf s,\lambda}\Phi_0 \right\|^2,
\end{equation}
where \(E(\tilde W_{\mathbf s})\) is the set of eigenvalues of \(\tilde W_{\mathbf s}\), and \(\mathcal P_{\mathbf s,\lambda}\) is the orthogonal projection onto the eigenspace associated with \(\lambda\).

\end{theorem}

\begin{proof}
For each \(t\in\mathbb N\),
\[
\tilde W_{\mathbf s}^t
=
\sum_{\lambda\in E(\tilde W_{\mathbf s})}
\lambda^t\, \mathcal P_{\mathbf s,\lambda}.
\]
Applying this to \(\Phi_0\) gives
\[
Q_\sigma \tilde W_{\mathbf s}^t \Phi_0
=
\sum_{\lambda\in E(\tilde W_{\mathbf s})}
\lambda^t\, Q_\sigma \mathcal P_{\mathbf s,\lambda}\Phi_0.
\]
Taking the squared norm and expanding the double sum, we get
\[
\left\| Q_\sigma \tilde W_{\mathbf s}^t \Phi_0 \right\|^2
=
\sum_{\lambda,\mu\in E(\tilde W_{\mathbf s})}
(\lambda\bar\mu)^t\,
\left\langle
Q_\sigma \mathcal P_{\mathbf s,\lambda}\Phi_0,\,
Q_\sigma \mathcal P_{\mathbf s,\mu}\Phi_0
\right\rangle.
\]
Since \(|\lambda|=|\mu|=1\), and
\[
\frac{1}{T}\sum_{t=0}^{T-1} (\lambda\bar\mu)^t
\longrightarrow
\begin{cases}
1, & \lambda=\mu,\\
0, & \lambda\neq \mu.
\end{cases}
\]
Therefore,
\[
\begin{aligned}
\bar P_\sigma(\Phi_0,\mathbf s)
&=
\lim_{T\to\infty}
\frac{1}{T}\sum_{t=0}^{T-1}
\left\| Q_\sigma \tilde W_{\mathbf s}^t \Phi_0 \right\|^2 \\
&=
\sum_{\lambda\in E(\tilde W_{\mathbf s})}
\left\| Q_\sigma \mathcal P_{\mathbf s,\lambda}\Phi_0 \right\|^2.
\end{aligned}
\]
\end{proof}

\begin{definition}[Disorder-free localization for a fixed spin configuration]
\label{def:DFL}
We say that disorder-free localization (DFL) occurs if there exist a spin configuration \(\mathbf s\in\{\pm1\}^{\Gamma_n}\) and a unit vector \(\Phi_0\in\mathfrak h_n\otimes\mathcal K\) such that the time-averaged position probability distribution
\[
\left(\bar P_\sigma(\Phi_0,\mathbf s)\right)_{\sigma\in\Gamma_n}
\]
is not uniform on \(\Gamma_n\), i.e.,
\[
\exists\,\mathbf s,\ \exists\,\Phi_0\ \text{with}\ \|\Phi_0\|=1,\ 
\exists\,\sigma,\sigma'\in\Gamma_n:
\bar P_\sigma(\Phi_0,\mathbf s)\neq \bar P_{\sigma'}(\Phi_0,\mathbf s).
\]
\end{definition}

\begin{remark}
\label{rem:DFL_definition}
The definition uses a fixed spin configuration \(\mathbf s\). This is justified because the total evolution \(\tilde W\) preserves each invariant sector \(\mathcal H_{\mathbf s}\). Hence any initial state in \(\mathcal H_{\mathbf s}\) remains in that sector at all times, and its long-time behavior is governed by the restriction \(\tilde W_{\mathbf s}\). Consequently, if there exist a spin configuration \(\mathbf s\) and a state in \(\mathcal H_{\mathbf s}\) whose time-averaged position probability is non-uniform, then DFL occurs. Thus, to demonstrate DFL it suffices to exhibit one spin configuration \(\mathbf s\) and one initial state in \(\mathcal H_{\mathbf s}\) for which the time-averaged position probability is not uniform.
\end{remark}

\begin{theorem}[General DFL criterion for a fixed spin configuration]
\label{thm:DFL_general}
Disorder-free localization occurs if and only if there exist a spin configuration \(\mathbf s\in\{\pm1\}^{\Gamma_n}\), an eigenvalue \(\lambda\in E(\tilde W_{\mathbf s})\), and a non-zero vector \(\Phi\in\operatorname{Ran}\mathcal P_{\mathbf s,\lambda}\) such that the position marginal probability distribution
\[
\left(\|Q_\sigma \Phi\|^2\right)_{\sigma\in\Gamma_n}
\]
is not constant. Equivalently,
\[
\exists\,\mathbf s\in\{\pm1\}^{\Gamma_n},\ 
\exists\,\lambda\in E(\tilde W_{\mathbf s}),\ 
\exists\,\Phi\in\operatorname{Ran}\mathcal P_{\mathbf s,\lambda}\setminus\{0\},\ 
\exists\,\sigma,\sigma'\in\Gamma_n:
\|Q_\sigma \Phi\|^2 \neq \|Q_{\sigma'} \Phi\|^2.
\]
\end{theorem}

\begin{proof}

\emph{Sufficiency.}
If there exist \(\mathbf s\), \(\lambda\), and \(\Phi\in\operatorname{Ran}\mathcal P_{\mathbf s,\lambda}\setminus\{0\}\) such that \(\{\|Q_\sigma\Phi\|^2\}_{\sigma}\) is not constant, then take \(\Phi_0=\Phi/\|\Phi\|\). Since \(\Phi\) lies in the \(\lambda\)-eigenspace, the sum in \eqref{eq:time_average_spec} reduces to the single term \(\lambda\), and hence
\[
\bar P_\sigma(\Phi_0,\mathbf s)
=
\left\| Q_\sigma \Phi_0 \right\|^2,
\]
which is non-uniform. Thus DFL occurs.

\emph{Necessity.}
Suppose DFL occurs. Then there exist \(\mathbf s\) and a unit vector \(\Phi_0\) such that \(\{\bar P_\sigma(\Phi_0,\mathbf s)\}_{\sigma}\) is not uniform. By \eqref{eq:time_average_spec}, if for every eigenvalue \(\lambda\in E(\tilde W_{\mathbf s})\) and every \(\sigma,\sigma'\in\Gamma_n\) we had
\[
\left\| Q_\sigma \mathcal P_{\mathbf s,\lambda}\Phi_0 \right\|^2
=
\left\| Q_{\sigma'} \mathcal P_{\mathbf s,\lambda}\Phi_0 \right\|^2,
\]
then the sum would also be uniform. Hence there must exist at least one eigenvalue \(\lambda\) and two vertices \(\sigma,\sigma'\) such that
\[
\left\| Q_\sigma \mathcal P_{\mathbf s,\lambda}\Phi_0 \right\|^2
\neq
\left\| Q_{\sigma'} \mathcal P_{\mathbf s,\lambda}\Phi_0 \right\|^2.
\]
Setting \(\Phi:=\mathcal P_{\mathbf s,\lambda}\Phi_0\), we obtain a non-zero vector in \(\operatorname{Ran}\mathcal P_{\mathbf s,\lambda}\) with non-uniform position marginal. This completes the proof.
\end{proof}

We now specialize the general DFL criterion to a constructive method. For a fixed spin configuration \(\mathbf s\), non-uniform vectors in the eigenvalue-one subspace \(\mathcal F_{\mathbf s}\) are characterized by a fixed-point equation, which we derive below and then solve explicitly for the Grover walk.

\begin{proposition}[Fixed-point equation]
\label{pro:fixed_point_equation}
Let \(\mathbf s\in\{\pm1\}^{\Gamma_n}\) be fixed, and let
\begin{equation}\label{eq:fixed_space_F_s}
\mathcal F_{\mathbf s} = \{\Phi\in\mathfrak h_n\otimes\mathcal K : \tilde W_{\mathbf s}\Phi = \Phi\}
\end{equation}
be the eigenvalue-one subspace of \(\tilde W_{\mathbf s}\). A unit vector
\[
\Phi = \sum_{\sigma\in\Gamma_n} Z_\sigma \otimes u_\sigma
\]
belongs to \(\mathcal F_{\mathbf s}\) if and only if the family \(\{u_\sigma\}_{\sigma\in\Gamma_n}\) satisfies
\begin{equation}\label{eq:fixed_point}
u_\sigma
=
\sum_{k=0}^{n}
C_k\,
e^{-i\phi_{\sigma\triangle\{k\}} s_{\sigma\triangle\{k\}}}\,
u_{\sigma\triangle\{k\}},
\qquad \forall \sigma\in\Gamma_n.
\end{equation}
\end{proposition}

\begin{proof}
By the definition of \(\tilde W_{\mathbf s}\), for any \(\tau\in\Gamma_n\), we have
\[
\tilde W_{\mathbf s}(Z_\tau\otimes u_\tau)
=
\sum_{k=0}^{n}
Z_{\tau\triangle\{k\}}
\otimes C_k\,
e^{-i\phi_\tau s_\tau}\, u_\tau.
\]
Summing over \(\tau\) gives
\[
\tilde W_{\mathbf s}\Phi
=
\sum_{\tau\in\Gamma_n}
\sum_{k=0}^{n}
Z_{\tau\triangle\{k\}}
\otimes C_k\,
e^{-i\phi_\tau s_\tau}\, u_\tau.
\]
Relabeling the position index as \(\sigma = \tau\triangle\{k\}\), we obtain \(\tau = \sigma\triangle\{k\}\), and the phase factor becomes
\[
e^{-i\phi_\tau s_\tau}
=
e^{-i\phi_{\sigma\triangle\{k\}} s_{\sigma\triangle\{k\}}}.
\]
Thus
\[
\tilde W_{\mathbf s}\Phi
=
\sum_{\sigma\in\Gamma_n}
Z_\sigma \otimes
\left(
\sum_{k=0}^{n}
C_k\,
e^{-i\phi_{\sigma\triangle\{k\}} s_{\sigma\triangle\{k\}}}\,
u_{\sigma\triangle\{k\}}
\right).
\]
Since \(\{Z_\sigma\}_{\sigma\in\Gamma_n}\) is an orthonormal basis, the equality \(\tilde W_{\mathbf s}\Phi = \Phi\) holds if and only if the coefficients of \(Z_\sigma\) coincide for every \(\sigma\). Therefore \(\Phi\in\mathcal F_{\mathbf s}\) is equivalent to
\[
u_\sigma
=
\sum_{k=0}^{n}
C_k\,
e^{-i\phi_{\sigma\triangle\{k\}} s_{\sigma\triangle\{k\}}}\,
u_{\sigma\triangle\{k\}},
\qquad \forall \sigma\in\Gamma_n,
\]
which completes the proof.
\end{proof}

\begin{corollary}
\label{cor:fixed_point_DFL}
Suppose that for some spin configuration $\mathbf s\in\{\pm1\}^{\Gamma_n}$, the fixed-point equation \eqref{eq:fixed_point} admits a non-zero solution $\Phi\in\mathcal F_{\mathbf s}$ whose position marginal probability distribution $\{\|Q_\sigma\Phi\|^2\}_{\sigma\in\Gamma_n}$ is not constant. Then DFL occurs.
\end{corollary}

\begin{proof}
Let $\Phi_0:=\Phi/\|\Phi\|$ be the normalized fixed point. Since $\Phi\in\mathcal F_{\mathbf s}$, we have $\mathcal P_{\mathbf s,\lambda}\Phi_0=0$ for all $\lambda\neq 1$, and $\mathcal P_{\mathbf s,1}\Phi_0=\Phi_0$. By \eqref{eq:time_average_spec} we obtain
\[
\bar P_\sigma(\Phi_0,\mathbf s)
=
\sum_{\lambda\in E(\tilde W_{\mathbf s})}
\left\| Q_\sigma\, \mathcal P_{\mathbf s,\lambda}\Phi_0 \right\|^2
=
\left\| Q_\sigma \Phi_0 \right\|^2
=
\frac{\|Q_\sigma \Phi\|^2}{\|\Phi\|^2}.
\]
Since the latter is non-uniform, the definition of DFL is satisfied.
\end{proof}

We now specialize to the general Grover coin, the operator \(C_k\) is given by
\[
C_k = \frac{2}{n+1} \sum_{j=0}^{n} |e_k\rangle \langle e_j| - |e_k\rangle \langle e_k|.
\]

And the coupling
\[
\phi_\sigma = \frac{|\sigma|}{n+1}\pi.
\]
where \(|\sigma|\) denotes the cardinality of \(\sigma\). Equivalently, \(|\sigma|\) is the Hamming distance from the empty set \(\emptyset\) to \(\sigma\).

Throughout the rest of this subsection, \(d = n+1\).

\subsection{Disorder-free localization in even dimensions}
\label{subsec:DFL_even}

We now present explicit non-uniform fixed points for the spin-coupled Grover walk in dimensions \(d=4\) and \(d=6\), establishing disorder-free localization in these cases. In both cases we take the all \(+1\) spin configuration \(s_\sigma=+1\) for all vertices \(\sigma\), and the distance-dependent coupling \(\phi_\sigma=|\sigma|\pi/d\).

\begin{proposition}[Non-uniform fixed point for \(d=4\)]
\label{prop:DFL_d4}
There exists a non-zero vector \(\Phi\in\mathfrak h_3\otimes\mathbb C^4\) satisfying the fixed-point equation
\[
\tilde W_{\mathbf s}\Phi=\Phi,
\]
where \(\tilde W_{\mathbf s}\) is the reduced evolution operator on the \(4\)-dimensional hypercube with Grover coin and coupling \(\phi_\sigma=|\sigma|\pi/4\), such that the position weights \(\|Q_\sigma\Phi\|^2\) are not constant.
\end{proposition}

\begin{proof}
Let \(d=4\) and denote the vertex set by \(\Gamma_3\). The fixed-point equation \eqref{eq:fixed_point} reads
\[
u_\sigma
=
\sum_{k=0}^{3}
C_k\,
e^{-i|\sigma\triangle\{k\}|\pi/4}\,
u_{\sigma\triangle\{k\}},
\qquad \forall\sigma\in\Gamma_3,
\]
where \(u_\sigma\in\mathbb C^4\) is the coin vector at vertex \(\sigma\). This is a system of \(64\) linear equations over the cyclotomic field \(\mathbb Q(e^{i\pi/4})\), which we solve symbolically. A non-zero solution is given in Table~\ref{tab:DFL_d4}.

\begin{table}[htbp]
\centering

\begin{tabular}{c c c c c c c}
\toprule
Vertex \(\sigma\) & \(|\sigma|\) & Weight &
\(u_\sigma^{(0)}\) & \(u_\sigma^{(1)}\) & \(u_\sigma^{(2)}\) & \(u_\sigma^{(3)}\) \\
\midrule
\(\emptyset\) & 0 & 0 & 0 & 0 & 0 & 0 \\
\(\{0\}\) & 1 & 2 & 0 & \(e^{3\pi i/4}\) & \(e^{-\pi i/4}\) & 0 \\
\(\{1\}\) & 1 & 2 & \(e^{3\pi i/4}\) & 0 & 0 & \(e^{-\pi i/4}\) \\
\(\{0,1\}\) & 2 & 4 & \(-i\) & \(-i\) & \(e^{5\pi i/4}\) & \(e^{5\pi i/4}\) \\
\(\{2\}\) & 1 & 2 & \(e^{-\pi i/4}\) & 0 & 0 & \(e^{3\pi i/4}\) \\
\(\{0,2\}\) & 2 & 4 & \(i\) & \(e^{\pi i/4}\) & \(i\) & \(e^{\pi i/4}\) \\
\(\{1,2\}\) & 2 & 0 & 0 & 0 & 0 & 0 \\
\(\{0,1,2\}\) & 3 & 2 & 0 & \(1\) & \(-1\) & 0 \\
\(\{3\}\) & 1 & 2 & 0 & \(e^{-\pi i/4}\) & \(e^{3\pi i/4}\) & 0 \\
\(\{0,3\}\) & 2 & 0 & 0 & 0 & 0 & 0 \\
\(\{1,3\}\) & 2 & 4 & \(e^{\pi i/4}\) & \(i\) & \(e^{\pi i/4}\) & \(i\) \\
\(\{0,1,3\}\) & 3 & 2 & \(1\) & 0 & 0 & \(-1\) \\
\(\{2,3\}\) & 2 & 4 & \(e^{-3\pi i/4}\) & \(e^{-3\pi i/4}\) & \(-i\) & \(-i\) \\
\(\{0,2,3\}\) & 3 & 2 & \(-1\) & 0 & 0 & \(1\) \\
\(\{1,2,3\}\) & 3 & 2 & 0 & \(-1\) & \(1\) & 0 \\
\(\{0,1,2,3\}\) & 4 & 0 & 0 & 0 & 0 & 0 \\
\bottomrule
\end{tabular}
\caption{Non-uniform unit-modulus fixed point for \(d=4\). All nonzero components are of the form \(e^{ik\pi/4}\). Vertices with weight \(0\) are omitted.}
\label{tab:DFL_d4}
\end{table}

A direct substitution into the fixed-point equation shows that the vector is indeed a fixed point. All nonzero components have modulus \(1\), hence the position weight \(\|u_\sigma\|^2\) equals the number of nonzero components. The table shows weights \(0,2,4\) with multiplicities \(4,8,4\), respectively, so the normalized position distribution is non-uniform (e.g.\ \(p_\emptyset=0\) while \(p_{\{0,1\}}=1/8\)). Thus DFL occurs.
\end{proof}

\begin{proposition}[Non-uniform fixed point for \(d=6\)]
\label{prop:DFL_d6}
There exists a non-zero vector \(\Phi\in\mathfrak h_5\otimes\mathbb C^6\) satisfying the fixed-point equation
\[
\tilde W_{\mathbf s}\Phi=\Phi,
\]
where \(\tilde W_{\mathbf s}\) is the reduced evolution operator on the \(6\)-dimensional hypercube with Grover coin and coupling \(\phi_\sigma=|\sigma|\pi/6\), such that the position weights \(\|Q_\sigma\Phi\|^2\) are not constant.
\end{proposition}

\begin{proof}
Let \(d=6\), vertex set \(\Gamma_5\), and let \(\zeta = e^{i\pi/6}\) be a 12th root of unity. The fixed-point equation is
\[
u_\sigma = \sum_{k=0}^{5} C_k\, e^{-i|\sigma\triangle\{k\}|\pi/6} u_{\sigma\triangle\{k\}},
\qquad \forall \sigma\in\Gamma_5.
\]
A non-zero solution is given in Table~\ref{tab:DFL_d6}; all components are expressed in terms of \(\zeta\). The table lists exactly those vertices with non-zero coin vector; all other vertices have \(u_\sigma=0\).

\begin{table}[htbp]
\centering
\begin{tabular}{c c c c c c c c c}
\toprule
Vertex \(\sigma\) & \(|\sigma|\) & Weight &
\(c_0\) & \(c_1\) & \(c_2\) & \(c_3\) & \(c_4\) & \(c_5\) \\
\midrule
(0,0,0,0,1,1) & 2 & 2 & \(1\) & 0 & 0 & \(-1\) & 0 & 0 \\
(0,0,0,1,0,1) & 2 & 2 & 0 & 0 & \(1\) & 0 & \(-1\) & 0 \\
(0,0,0,1,1,0) & 2 & 2 & 0 & \(1\) & 0 & 0 & 0 & \(-1\) \\
(0,0,1,0,0,1) & 2 & 2 & \(-1\) & 0 & 0 & \(1\) & 0 & 0 \\
(0,0,1,1,0,0) & 2 & 2 & 0 & \(-1\) & 0 & 0 & 0 & \(1\) \\
(0,1,0,0,1,0) & 2 & 2 & \(-1\) & 0 & 0 & \(1\) & 0 & 0 \\
(0,1,0,1,0,0) & 2 & 2 & 0 & 0 & \(-1\) & 0 & \(1\) & 0 \\
(0,1,1,0,0,0) & 2 & 2 & \(1\) & 0 & 0 & \(-1\) & 0 & 0 \\
(1,0,0,0,0,1) & 2 & 2 & 0 & 0 & \(-1\) & 0 & \(1\) & 0 \\
(1,0,0,0,1,0) & 2 & 2 & 0 & \(-1\) & 0 & 0 & 0 & \(1\) \\
(1,0,1,0,0,0) & 2 & 2 & 0 & \(1\) & 0 & 0 & 0 & \(-1\) \\
(1,1,0,0,0,0) & 2 & 2 & 0 & 0 & \(1\) & 0 & \(-1\) & 0 \\
\addlinespace
(0,0,1,1,1,1) & 4 & 2 & 0 & 0 & \(-\zeta\) & 0 & \(\zeta\) & 0 \\
(0,1,0,1,1,1) & 4 & 2 & 0 & \(-\zeta\) & 0 & 0 & 0 & \(\zeta\) \\
(0,1,1,1,0,1) & 4 & 2 & 0 & \(\zeta\) & 0 & 0 & 0 & \(-\zeta\) \\
(0,1,1,1,1,0) & 4 & 2 & 0 & 0 & \(\zeta\) & 0 & \(-\zeta\) & 0 \\
(1,0,0,1,1,1) & 4 & 2 & \(-\zeta\) & 0 & 0 & \(\zeta\) & 0 & 0 \\
(1,0,1,0,1,1) & 4 & 2 & 0 & 0 & \(\zeta\) & 0 & \(-\zeta\) & 0 \\
(1,0,1,1,0,1) & 4 & 2 & \(\zeta\) & 0 & 0 & \(-\zeta\) & 0 & 0 \\
(1,1,0,0,1,1) & 4 & 2 & 0 & \(\zeta\) & 0 & 0 & 0 & \(-\zeta\) \\
(1,1,0,1,1,0) & 4 & 2 & \(\zeta\) & 0 & 0 & \(-\zeta\) & 0 & 0 \\
(1,1,1,0,0,1) & 4 & 2 & 0 & \(-\zeta\) & 0 & 0 & 0 & \(\zeta\) \\
(1,1,1,0,1,0) & 4 & 2 & 0 & 0 & \(-\zeta\) & 0 & \(\zeta\) & 0 \\
(1,1,1,1,0,0) & 4 & 2 & \(-\zeta\) & 0 & 0 & \(\zeta\) & 0 & 0 \\
\addlinespace
(0,0,0,1,1,1) & 3 & 6 & \(-\zeta^3\) & \(-\zeta^3\) & \(-\zeta^3\) & \(1-\zeta^2\) & \(1-\zeta^2\) & \(1-\zeta^2\) \\
(0,0,1,1,0,1) & 3 & 6 & \(\zeta^3\) & \(\zeta^3\) & \(-1+\zeta^2\) & \(-1+\zeta^2\) & \(\zeta^3\) & \(-1+\zeta^2\) \\
(0,1,0,1,1,0) & 3 & 6 & \(\zeta^3\) & \(-1+\zeta^2\) & \(\zeta^3\) & \(-1+\zeta^2\) & \(-1+\zeta^2\) & \(\zeta^3\) \\
(1,0,0,0,1,1) & 3 & 6 & \(-1+\zeta^2\) & \(\zeta^3\) & \(\zeta^3\) & \(\zeta^3\) & \(-1+\zeta^2\) & \(-1+\zeta^2\) \\
(1,0,1,0,0,1) & 3 & 6 & \(1-\zeta^2\) & \(-\zeta^3\) & \(1-\zeta^2\) & \(-\zeta^3\) & \(-\zeta^3\) & \(1-\zeta^2\) \\
(1,1,0,0,1,0) & 3 & 6 & \(1-\zeta^2\) & \(1-\zeta^2\) & \(-\zeta^3\) & \(-\zeta^3\) & \(1-\zeta^2\) & \(-\zeta^3\) \\
(1,1,1,0,0,0) & 3 & 6 & \(-1+\zeta^2\) & \(-1+\zeta^2\) & \(-1+\zeta^2\) & \(\zeta^3\) & \(\zeta^3\) & \(\zeta^3\) \\
\bottomrule
\end{tabular}
\caption{Non-uniform unit-modulus fixed point for \(d=6\). All nonzero components are \(\zeta^k\), \(\zeta=e^{i\pi/6}\). Vertices with weight \(0\) are omitted.}
\label{tab:DFL_d6}
\end{table}

Direct substitution verifies that this vector satisfies the fixed-point equation exactly (symbolic computation in \(\mathbb Q(\zeta)\)). All nonzero components have modulus \(1\), so the position weight equals the number of non-zero entries. There are \(20\) vertices with weight \(2\) and \(4\) vertices with weight \(6\), all other vertices have weight \(0\). Hence the position distribution is non-uniform (e.g.\ some vertices have probability \(0\), others positive), proving DFL.
\end{proof}

\begin{remark}
In both \(d=4\) and \(d=6\), the constructed fixed points lie in the standard representation isotypic component of the \(S_d\)-action on the position-coin space. Numerical projection shows that the entire fixed-point space is contained in this component (dimension \(6\) for \(d=4\), \(20\) for \(d=6\)). This observation suggests that non-uniform fixed points are generally associated with the standard representation.
\end{remark}

\begin{conjecture}[Disorder-free localization in all even dimensions]
\label{conj:DFL_general}
For every even \(d\ge 4\), the spin-coupled Grover walk on the \(d\)-dimensional hypercube with coupling \(\phi_\sigma=|\sigma|\pi/d\) and all \(+1\) spin configuration exhibits disorder-free localization. Equivalently, the reduced evolution operator \(\tilde W_{\mathbf s}\) possesses a fixed point whose position distribution is non-uniform.
\end{conjecture}

Propositions \ref{prop:DFL_d4} and \ref{prop:DFL_d6} establish the conjecture for \(d=4\) and \(d=6\). For \(d=8\), numerical calculations show a fixed-point space of dimension \(70\), but an explicit unit-modulus sparse fixed point has not yet been found; the standard representation isotypic component remains the natural candidate. No non-uniform fixed points have been found for odd \(d\) or for \(d=2\), supporting the belief that DFL occurs if and only if \(d\ge4\) is even.

\section{Conclusion Remark}

We have studied a spin-coupled quantum walk on the hypercube, where the position-dependent spin rotation provides a mechanism that can produce DFL without external randomness. By decomposing the total evolution into fixed spin-configuration sectors, we reduced the long-time position probability distribution to the spectral subspaces of a family of unitary operators \(\tilde W_{\mathbf s}\). In addition, we obtained an exact path-integral representation for the finite-time position probability in each spin sector, expressing it as a sum over pairs of paths with spin-dependent phase factors. We then derived a general criterion for disorder-free localization within our model: DFL occurs if and only if there exists a spin configuration \(\mathbf s\) and an eigenvalue \(\lambda\) of \(\tilde W_{\mathbf s}\) such that the corresponding eigenspace contains a vector with a non-uniform position marginal. We specialized this criterion to the Grover walk with distance-dependent coupling \(\phi_\sigma=|\sigma|\pi/(n+1)\). For dimension \(d=4\), we explicitly solved the fixed-point equation in the all \(+1\) spin sector and obtained a non-uniform fixed point, thereby giving a rigorous proof that disorder-free localization occurs in this specific model.

Several questions remain open. The explicit construction for \(d=4\), together with numerical computations for \(d=6\) and \(d=8\), strongly suggests that DFL occurs for all even dimensions \(d\ge 4\) in the present Grover walk with distance-dependent coupling and the all \(+1\) spin configuration. However, a complete proof for general even dimensions is still lacking. For odd dimensions, the absence of DFL has not been established in full generality. Our fixed-point analysis only excludes non-uniform vectors in the eigenvalue-one subspace \(\mathcal F_{\mathbf s}\); however, the general DFL criterion also allows contributions from other eigenspaces. A complete proof of no DFL in odd dimensions would require excluding non-uniform position marginals in every eigenspace of \(\tilde W_{\mathbf s}\).

In summary, this work clarifies a boundary of disorder-free localization: while translational invariance alone may preclude DFL in simple settings, the presence of a structured spin environment can generate an effective position-dependent phase, thereby breaking translational symmetry within each spin sector and inducing localization even on a regular graph, as demonstrated by the explicit \(d=4\) construction.

\section*{Author declarations}

\subsection*{Conflict of interest}

The author has no conflicts to disclose.

\subsection*{Data Availability}

Data sharing is not applicable to this article as no new data were created or analyzed in this study.

\end{document}